\documentclass[12pt]{article}
\usepackage{amsmath,amsthm,amsfonts,amssymb,epsfig}
\usepackage{color}

\newcommand{\reals}{\mathbb{R}}

\newcommand{\exponent}{\operatorname{e}}

\newtheorem{stel}{Theorem}
\newtheorem{gevolg}{Corollary}
\newtheorem{lemma}{Lemma}
\theoremstyle{remark}
\newtheorem{opm}{Remark}

\begin{document}

\title{Tragedy of the Commons in the Chemostat}

\author{Martin Schuster\footnote{Department of Microbiology, Oregon State University, Supported by NSF-MCB-1158553, Martin.Schuster@oregonstate.edu}, Eric Foxall\footnote{School of Mathematical \& Statistical Sciences, Arizona State University, Eric.Foxall@asu.edu}, David Finch\footnote{Department of Mathematics, Oregon State University, finch@math.oregonstate.edu}, Hal Smith\footnote{School of Mathematical \& Statistical Sciences, Arizona State University, Supported by Simons Foundation Grant 355819, halsmith@asu.edu}, and Patrick De Leenheer\footnote{Department of Mathematics and Department of Integrative Biology, Oregon State University, Supported in part by NSF-DMS-1411853, deleenhp@math.oregonstate.edu}}

\maketitle

\begin{abstract}
We present a proof of principle for the phenomenon of the tragedy of the commons
that is at the center of many theories on the evolution of cooperation. We establish the tragedy in the context of a general chemostat model with two species, the cooperator and the cheater. Both species have the same growth rate function and yield constant, but the cooperator allocates a portion of the nutrient uptake towards the production of a public good -the ``Commons" in the Tragedy- which is needed to digest the externally supplied nutrient. The cheater on the other hand does not produce this enzyme, and allocates all nutrient uptake towards its own growth. We prove that when the cheater is present initially, both the cooperator and the cheater will eventually go extinct, hereby confirming the occurrence of the tragedy. We also show that without the cheater, the cooperator can survive indefinitely, provided that at least a low level of public good or processed nutrient is available initially. 
Our results provide a predictive framework for the analysis of cooperator-cheater dynamics in a powerful model system of experimental evolution.
\end{abstract}

\section{Introduction}
Cooperative behaviors abound across all domains of life, from animals to microbes \cite{frank,west1}.  Yet, their evolution and maintenance is difficult to explain \cite{lehman,west2,west3}. Why would an individual carry out a costly behavior for the benefit of the group?  Cheaters that reap the benefits of cooperation without paying the costs would gain a competitive advantage and invade the population. This conflict of interest between the individual and the group is also known as the Òtragedy of the commonsÓ described by Hardin \cite{hardin}.  To illustrate the tragedy, Hardin considers a scenario first sketched by Lloyd more than 100 years earlier \cite{lloyd}, a pasture that is shared by herdsmen. It is in each herdsman's best interest to add additional cattle to the pasture, because he gains the profits from individual cattle sales, but shares the costs of overgrazing with all other herdsmen. This behavior is pursued until, ultimately, the commons is destroyed to the detriment of all.

The problem of cooperation has received considerable attention in the microbial realm \cite{west1,asfahl,foster}. Many microbes perform cooperative behaviors such as biofilm formation, virulence, and collective nutrient acquisition.  Often, these behaviors are accomplished by secreted products referred to as public goods \cite{west1,west3}.  Public goods are costly to produce for the individual but provide a collective benefit to the local group.  They include extracellular enzymes that degrade complex food sources, siderophores that scavenge iron from the environment, and secreted toxins and antibiotics that harm other cells.  It has been shown in several microbial systems that public goods can be shared within a population of cells, benefitting cells other than the focal producer \cite{diggle,griffin,rainey,sandoz,greig}. For example, when the bacterium {\it Pseudomonas aeruginosa} is grown on a proteinaceous substrate, mutants deficient in protease secretion enrich in co-culture with the wild-type parent \cite{diggle,sandoz}.  These non-producing strains are termed obligate cheaters: They cannot grow by themselves, but they have a relative growth advantage in mixed cultures with cooperators.  Because cheater enrichment inevitably imposes a burden on the population, the expected outcome is a collapse of the population \cite{rankin}.  This outcome been shown experimentally in a few cases
\cite{rainey,fiegna,dandekar}. Often, however, cooperative behaviors are stably maintained and hence, the focus has largely been on mechanisms that avoid a tragedy of the commons \cite{asfahl2,chuang,fiegna2,foster2,gore,kummerli,waite,wang,xavier}.

To our knowledge, the notion that obligate cheating behavior constitutes a tragedy of the commons and leads to population collapse has not been mathematically proven.  Here, we consider the dynamics between cooperators and obligate cheaters in a continuous culture system.  Continuous cultures or chemostats enable microbial culturing at a specified density and growth rate through the constant dilution of the culture with fresh growth medium \cite{chemostat}.  There is an extensive mathematical theory that describes population dynamics in the chemostat \cite{chemostat}.  We prove that obligate cheaters inevitably increase in frequency until cooperation via public goods is no longer sustainable, eventually leading to wash-out and population collapse. We also show that the dynamics of the cooperators in the absence of cheaters exhibits bistability: Depending on the initial condition of the system, cooperators will either eventually persist, or go extinct. Numerical simulations show that it is possible that the cooperators persist when initially there is no processed nutrient, and only a small
level of enzyme. In summary, populations solely comprised of cooperators have a chance to persist, but they are doomed whenever cheaters arise, even at low initial frequency.



\section{The model and the tragedy}
We propose a chemostat model where
$S$ denotes the concentration of the unprocessed nutrient, $P$ of the processed nutrient, $E$ of the enzyme and $X_1$ is the
concentration of the cooperator who produces an enzyme required for nutrient processing, and $X_2$ of the cheater who does not produce the enzyme. The mass-balance equations for these variables are as follows:
\begin{eqnarray}
\frac{dS}{dt}(t)&=&D(t)(S^0(t)-S)-G(E,S)\label{s1}\\
\frac{dP}{dt}(t)&=&G(E,S)-\frac{1}{\gamma}\left(X_1+X_2 \right)F(P)-D(t)P\label{s2}\\
\frac{dE}{dt}(t)&=&(1-q)X_1F(P)-D(t)E\label{s3}\\
\frac{dX_1}{dt}(t)&=&X_1\left(qF(P)-D(t) \right)\label{s4}\\
\frac{dX_2}{dt}(t)&=&X_2\left(F(P)-D(t)\right)\label{s5}
\end{eqnarray}
The operating conditions of the chemostat may fluctuate in time, and they
are characterized by $D(t)$, the dilution rate, and $S^0(t)$, the concentration of the unprocessed nutrient at the inflow.
Both are non-negative functions of time, and additional assumptions for these functions will be introduced below.
Unprocessed nutrient is converted into processed nutrient by means of the enzyme. Processed nutrient is produced at rate $G(E,S)$.
The per capita consumption rate  of processed nutrient by both species is the same, and denoted by $\frac{1}{\gamma}F(P)$, where
$\gamma$ is the yield of this process, taking a value in $(0,1)$, and which is also assumed to be the same for both species.
The cooperator allocates a proportion $q$, a fixed value in $(0,1)$, of the processed nutrient it has consumed, towards its own growth. The remaining fraction $(1-q)$ goes towards the production of the enzyme which is needed to process the unprocessed nutrient. The cheater allocates all processed nutrient it has taken up towards growth. A cartoon of this chemostat model is presented in Figure $\ref{cartoon}$.
\begin{figure}[h]
\begin{center}
\includegraphics[height=4.5in]{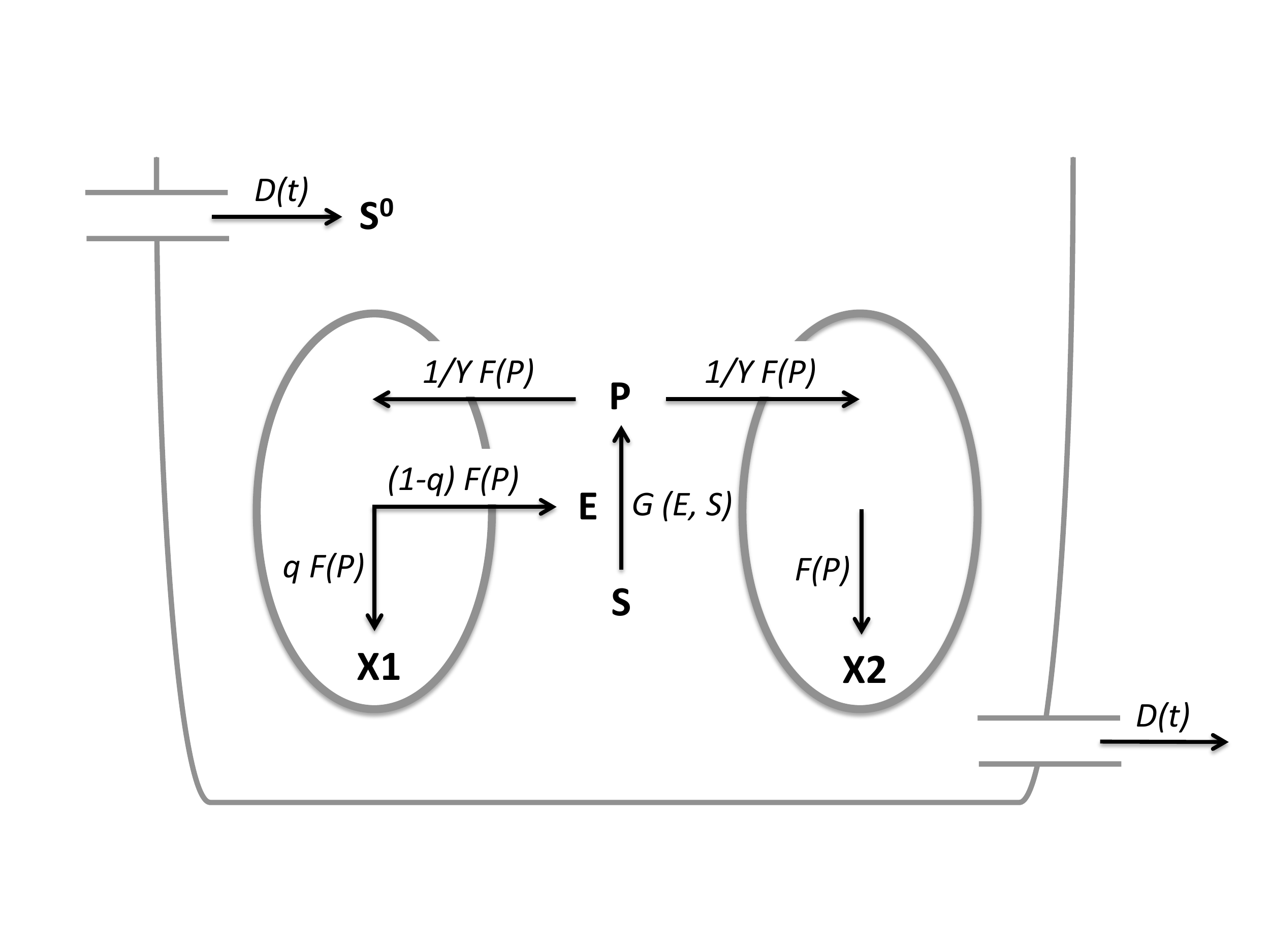}
\caption{Cartoon of the chemostat with two competing bacterial types: Species are indicated in bold and rates are indicated in italics. $X_1$, cooperator; $X_2$, cheater; $S$, nutrient substrate; $S^0$, unprocessed nutrient substrate in inflow; $P$, processed nutrient; $E$, enzyme; $D(t)$, dilution rate, $1/y F(P)$, per capita nutrient consumption rate; $F(P)$, growth rate; $q$ and $1-q$, proportions of nutrient allocated towards growth and enzyme production, respectively.}
\label{cartoon}
\end{center}
\end{figure}

We make the following minimal assumptions about the functions $G$ and $F$:
\begin{center}
{\bf H1}: $G:\reals_+\times \reals_+\to \reals_+$ is $C^1$, $G(0,S)=G(E,0)=0$ for all $E\geq 0$ and $S\geq 0$, and\\
$F:\reals_+\to \reals_+$ is $C^1$, and $F(0)=0$.
\end{center}
This assumption merely implies that there is no conversion of unprocessed nutrient into processed nutrient, when the enzyme or the unprocessed nutrient is missing; similarly there is no growth of either species, or of the enzyme, when the processed nutrient is missing.

For the dilution rate $D(t)$, and input nutrient concentration $S^0(t)$, we assume the following:
\begin{center}
{\bf H2}: The functions $D(t)$ and $S^0(t)$ are continuous for all $t\geq 0$, and
there exist positive bounds ${\underline D}$ and ${\bar D}$ such that ${\underline D}\leq D(t)\leq {\bar D}$ for all $t\geq 0$, and\\
positive bounds ${\underline S^0}$ and ${\bar S^0}$ such that ${\underline S^0}\leq S^0(t)\leq {\bar S^0}$ for all $t\geq 0$.
\end{center}

Our Main Result, which is proved in the Appendix, establishes the tragedy of the commons:
\begin{stel}\label{main} Assume that {\bf H1} and {\bf H2} hold, and assume that
the initial condition of $(\ref{s1})-(\ref{s5})$ is such that $X_2(0)>0$; that is, the cheater is present initially. Then
$(P(t),E(t),X_1(t),X_2(t))\to (0,0,0,0)$ as $t\to \infty$.
\end{stel}

Figure~\ref{fig:1} depicts the tragedy in case of  mass action kinetics $G(E,S)=kES$, and Monod
uptake function $F(P)=mP/(a+P)$. The equations have been scaled such that $S^0$ and $D$ are both constant equal to one.
Initial data are as follows: $S(0)=1, P(0)=0, E(0)=0.8, X_1(0)=0.2, X_2(0)=0.03$.
The cooperator peaks early and declines sharply as the cheater continues to thrive, reaching a maximum followed by a rapid decline.

\begin{figure}[h]
\begin{center}
\includegraphics[height=4.75in]{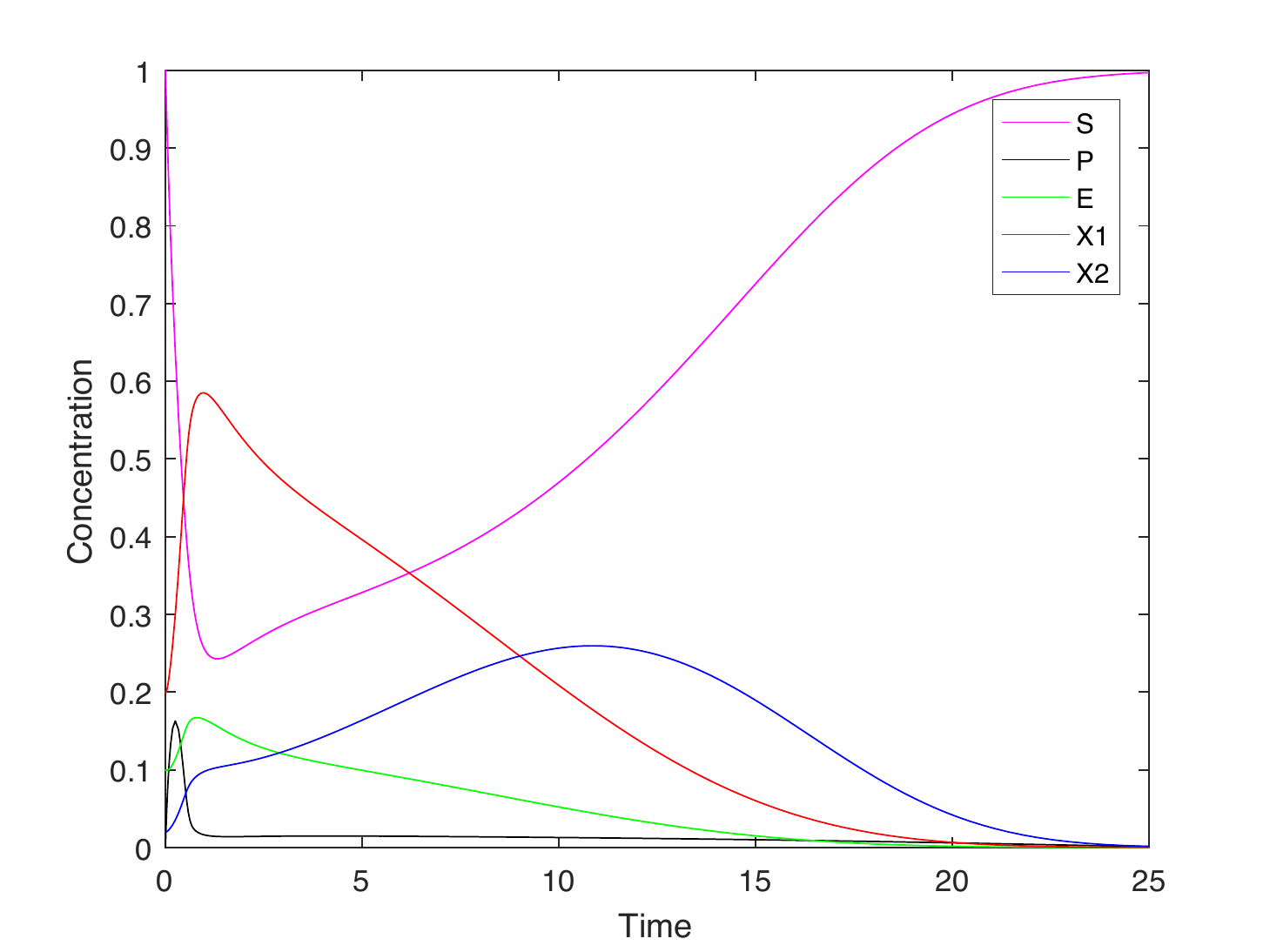}
\caption{Time series of the components of system $(\ref{s1})-(\ref{s5})$, where $S^0=1$, $D=1$, $q=0.8$, $\gamma=1$, $G(E,S)=kES$, $F(P)=mP/(a+P)$ with
$k=20$, $m=5$ and $a=0.05$. Initial data: $S(0)=1, P(0)=0, E(0)=0.1, X_1(0)=0.2, X_2(0)=0.02$.}
\label{fig:1}
\end{center}
\end{figure}

We show next that the tragedy also occurs in cases where the processing of the substrate into processed nutrient proceeds in more than one step. First, let us single out the biochemical reaction taking place in model $(\ref{s1})-(\ref{s5})$. Borrowing notation from (bio)chemistry, this
reaction can be represented as follows:
\begin{equation*}\label{biochem1}
S+E\to P+E,
\end{equation*}
where the reaction rate of formation of processed nutrient is $g(e,s)$, expressed in rescaled variables (see the Appendix for the rescaling). If we would only model this process, and ignore enzyme production, inflow of substrate, and outflow of substrate, enzyme and processed nutrient, we would have the following mass balance:
\begin{eqnarray*}
\frac{ds}{dt}(t)&=&-g(e,s)\\
\frac{de}{dt}(t)&=&0\\
\frac{dp}{dt}(t)&=&g(e,s)
\end{eqnarray*}

Suppose now that the biochemistry describing the conversion of substrate into processed nutrient takes occurs via an intermediate step:
\begin{equation*}\label{biochem2}
S+E\longleftrightarrow C\to P+E,
\end{equation*}
where $C$ represents an intermediate complex formed by the action of the enzyme on the substrate. Let us for simplicity assume that the reaction rates are of the mass action type (with respective rate constants $k_1$ and $k_{-1}$ for the first reversible reaction, and $k_2$ for the second reaction), then the mass balance model for this 2-step biochemical reaction network is:
\begin{eqnarray*}
\frac{ds}{dt}(t)&=&-k_1es+k_{-1}c\\
\frac{de}{dt}(t)&=&-k_1es+k_{-1}c+k_2c\\
\frac{dc}{dt}(t)&=&k_1es-k_{-1}c-k_2c\\
\frac{dp}{dt}(t)&=&k_2c
\end{eqnarray*}
The key property for this network is the conservation of the following quantity:
$$
s(t)+e(t)+2c(t)+p(t),
$$
which is easily verified by showing that its derivative with respect to time is zero. If we integrate this biochemical reaction network in our chemostat model, then we obtain the following scaled chemostat model:
\begin{eqnarray}
\frac{ds}{dt}(t)&=&D(t)(S^0(t)-s)-k_1es+k_{-1}c\label{bio1}\\
\frac{dp}{dt}(t)&=&k_2c-\left(x_1+x_2 \right)f(p)-D(t)p\label{bio2}\\
\frac{de}{dt}(t)&=&(1-q)x_1f(p)-k_1es+k_{-1}c+k_2c-D(t)e\label{bio3}\\
\frac{dc}{dt}(t)&=&k_1es-k_{-1}c-k_2c-D(t)c\label{bio4}\\
\frac{dx_1}{dt}(t)&=&x_1\left(qf(p)-D(t) \right)\label{bio5}\\
\frac{dx_2}{dt}(t)&=&x_2\left(f(p)-D(t)\right)\label{bio6}
\end{eqnarray}
We show in the last section of the Appendix that the tragedy continues to hold, in the sense that the conclusion of Theorem $\ref{main}$ remains valid for this more general system.

Of course, more complicated biochemical reaction networks of the digestion process, with multiple intermediate complexes $C_1,\dots C_n$:
$$
S+E \longleftrightarrow C_1 \longleftrightarrow \dots \longleftrightarrow C_n \to P+E
$$
could be used here instead, and the tragedy would continue to hold in such cases. The key property is that the mass balance equations corresponding to these networks should
exhibit a conservation law to guarantee the boundedness of the solutions of the chemostat model which integrates this biochemistry. Most
reasonable biochemical reaction networks do indeed possess such conservation laws.

\section{Cooperators can persist when cheaters are absent}
We have shown that when cheaters are present initially, the total population of cooperators and cheaters, is doomed. Next we investigate what happens when cheaters are absent by considering a special case of the chemostat model $(\ref{s1})-(\ref{s5})$ with $X_2=0$,
and constant operating parameters $D$ and $S^0$, which are both assumed to be positive:
\begin{eqnarray}
\frac{dS}{dt}(t)&=&D(S^0-S)-EG(S)\label{c1}\\
\frac{dP}{dt}(t)&=&EG(S)-\frac{1}{\gamma}X_1F(P)-DP\label{c2}\\
\frac{dE}{dt}(t)&=&(1-q)X_1F(P)-DE\label{c3}\\
\frac{dX_1}{dt}(t)&=&X_1\left(qF(P)-D \right) \label{c4}
\end{eqnarray}
Notice that the nutrient processing rate has been specialized to $EG(S)$, implying that it is proportional to the enzyme concentration $E$, and
a possibly nonlinear function of the nutrient $G(S)$. We replace assumption {\bf H1}, by the following assumption, which introduces
a monotonicity condition for $F$, and monotonicity and
concavity condition for $G$:
\begin{center}
{\bf H1'}: $G: \reals_+\to \reals_+$ is  $C^2$, $G(0)=0$, $dG/dS(S)>0$ for all $S>0$, and \\
$d^2G/dS^2(S)\leq 0$ for all $S\geq 0$, and \\
$F:\reals_+\to \reals_+$ is $C^1$, $F(0)=0$, $dF/dP(P)>0$ for all $P>0$.
\end{center}
The concavity condition for $G$ will be used to limit the number of
steady states of this system. The most commonly used choices for the functions for $F$ and $G$ are Monod functions (i.e. $F(P)=mP/(a+P)$, where $a$ and $m$ are positive parameters), which satisfy these assumptions. But note that a linear function $G(S)=kS$, with $k>0$ is allowed as well. In other words, the processing rate of nutrient (per unit of enzyme) does not necessarily have to saturate for large $S$-values.

The following dichotomy -global extinction, or bistability- is proved in the Appendix, and shows that the cooperator may persist when there are no cheaters; it refers to a scalar, nonlinear equation $(\ref{intersect})$, which is given in the Appendix as well.
\begin{stel}\label{bistab-unscaled}
Suppose that {\bf H1'} holds, and that $P^*:=F^{-1}\left(\frac{D}{q}\right)<S^0$.
\begin{enumerate}
\item
If equation $(\ref{intersect})$ has no solutions, then the washout steady state $(0,0,0,0)$ is globally asymptotically stable for
system $(\ref{c1})-(\ref{c4})$.
\item
If equation $(\ref{intersect})$ has two distinct  solutions, then
system $(\ref{c1})-(\ref{c3})$ has 3 steady states, the washout steady state $(0,0,0,0)$ and two positive steady states $E_1$ and $E_2$. The washout steady
state and $E_2$ are locally asymptotically stable, and $E_1$ is a saddle with a  three-dimensional stable manifold, and one-dimensional unstable manifold. The stable manifold is the common boundary of the regions of attraction of the washout steady state and $E_2$. Every solution of system $(\ref{c1})-(\ref{c4})$ converges to one of the three steady states. Persistence of the cooperator occurs for all initial conditions contained in the region of attraction of $E_2$, and initial conditions on the stable manifold of the saddle $E_1$.
\end{enumerate}
\end{stel}
Figure $\ref{fig:3}$ illustrates the persistence of the cooperator in the absence of cheaters, even when there is no 
processed nutrient, and only a little amount of enzyme initially.  
Notice that the initial condition used in the simulation for Figure $\ref{fig:3}$ is the same as the initial condition used for Figure $\ref{fig:1}$, and the model parameters are the same as well. Nevertheless, the fate of the cooperator is very different: it goes extinct when the cheater is present initially (Figure $\ref{fig:1}$), but persists otherwise (Figure $\ref{fig:3}$).

\section{Conclusion}
Although the tragedy of the commons is such a pervasive notion in the recent developments of theories about the evolution of cooperation,
we were unable to find any mathematical models that have rigorously analyzed an important group-level effect: the collapse of a population 
as a consequence of the dynamic interaction between cooperating and cheating individuals.
Here we have proved mathematically that the tragedy of the commons occurs in a chemostat system with cooperators that supply
a public good required for growth, and cheaters that do not. The sole difference between cooperators and cheaters in this system is the 
cost associated with public good production, which is only experienced by the cooperator. While the cooperator diverts a fraction of the 
ingested nutrient from growth to public good production, the cheater invests everything in growth. We assume that there are no pleiotropic 
costs to  cheating, and that the environment is well mixed, disregarding spatial 
structure as a major factor that promotes cooperation \cite{west2,west3}. Our results support the occurrence of the tragedy of the commons 
as a consequence of the selfish actions of individuals that result in the complete collapse of the shared public good \cite{hardin,rankin}. When this public good is essential for growth, the tragedy is manifested by the extinction of the whole group \cite{rainey,fiegna2,dandekar}.

To understand how the tragedy of the commons arises in the chemostat, we perform a simple thought-experiment. Assume that initially there are no cheaters ($X_2(0)=0$), and suppose that the assumptions of Theorem $\ref{bistab-unscaled}$ hold. If the initial condition of system $(\ref{c1})-(\ref{c4})$ is contained in the region of attraction of the locally stable steady state $E_2$, the solution will converge to, and eventually settles at this steady state.
Numerical simulations (Figure~\ref{fig:3})  show that this can happen even if there is only a  low initial amount of enzyme ($E(0)$ is small), and no initial processed nutrient ($P(0)=0$).  The cooperator-only population therefore
persists. However, if cheaters do suddenly appear -for example by mutation or by invasion into the environment-
even in extremely low numbers, Theorem $\ref{main}$ shows that the total population of cooperators and cheaters is doomed, confirming the tragedy of the commons. One of the two proofs of Theorem $\ref{main}$ gives clues on how this happens: The {\it ratio of cooperators to cheaters} will always decrease.
It may appear as if the cheaters will overtake the cooperators, and at least for a while, this is indeed what happens. 
However, in the long run there are not enough cooperators around to produce the enzyme levels required for nutrient processing, and this leads to the extinction of cheaters and cooperators alike.

\begin{figure}[h]
\begin{center}
\includegraphics[height=4.75in]{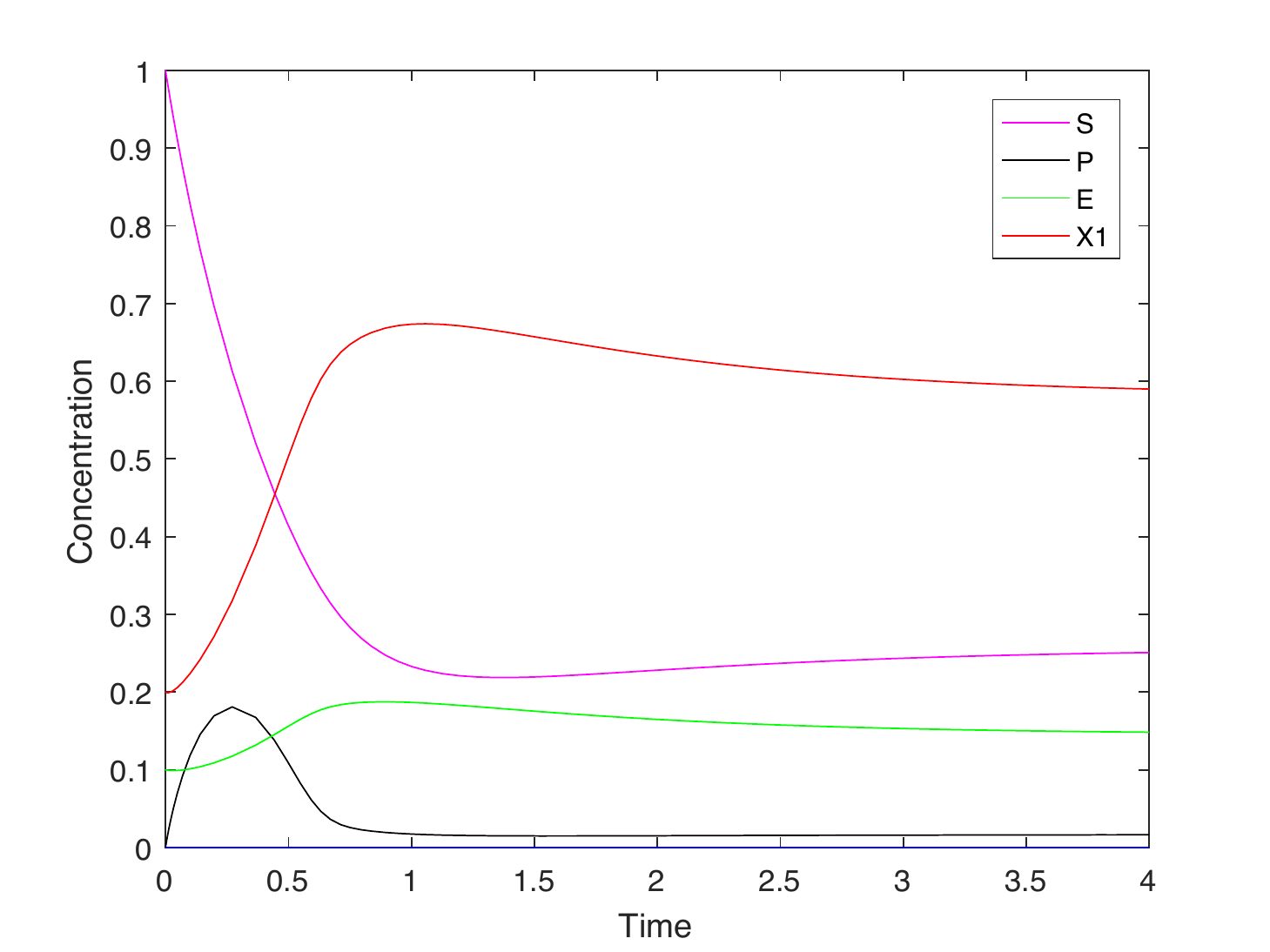}
\caption{Time series of the components of system $(\ref{c1})-(\ref{c4})$, where $S^0=1$, $D=1$, $q=0.8$, $\gamma=1$, $G(E,S)=kES$, $F(P)=mP/(a+P)$ with
$k=20$, $m=5$ and $a=0.05$. Initial data: $S(0)=1, P(0)=0, E(0)=0.1, X_1(0)=0.2$.}
\label{fig:3}
\end{center}
\end{figure}


To put our results in context, it is useful to quote from Hardin's original interpretation of the tragedy, see \cite{hardin}:

{\it The tragedy of the commons develops in this way. Picture a pasture open to all. It is to be expected that each herdsman will try to keep as many cattle as possible on the commons. Such an arrangement may work reasonably satisfactorily for centuries because tribal
wars, poaching, and disease keep the numbers of both man and beast well below the carrying capacity of the land. Finally, however, comes the day of reckoning, that is, the day when the long-desired goal of social stability becomes a reality. At this point, the inherent logic of the commons remorselessly generates tragedy.

As a rational being, each herdsman seeks to maximize his gain. Explicitly or implicitly, more or less consciously,  he asks, "What is the utility to me of adding one more animal to my herd?" This utility has one negative and one positive component:
\begin{enumerate}
\item
The positive component is a function of the increment of one animal. Since the herdsman receives all the proceeds from the sale of the additional animal, the positive utility is nearly $+1$.
\item
The negative component is a function of the additional overgrazing created by one or more animal. Since, however, the effects of overgrazing are shared by all herdsman, the negative utility for any particular decision making herdsman is only a fraction of $-1$.
\end{enumerate}

Adding together the component partial utilities, the rational herdsman concludes that the only sensible course
for him to pursue is to add another animal to his herd. And another; and another.... But this is the conclusion reached by each and every rational herdsman sharing a commons. Therein is the tragedy. Each man is locked into a system that compels him to increase his herd without limit-in a world that is limited. Ruin is the destination toward which all men rush, each pursuing his own best interest in a society that
believes in the freedom of the commons. Freedom in a commons brings ruin to all. }

It is interesting to note that Hardin's verbal description of the tragedy makes no explicit distinction between cooperators and cheaters, which is in contrast with recent interpretations of the 
tragedy in evolutionary biology \cite{rankin}. In natural populations there are many different ways individuals can cooperate or cheat, and clearly articulating the distinction between both types is necessary to correlate it to the occurrence of the tragedy \cite{rankin}. In its essence, the tragedy of the commons is the depletion of a common resource or a public good by the selfish action of competing individuals, thereby decreasing the average fitness of all individuals.

In the realm of game theory, the tragedy is described by a public good game or an $N$-person prisoner's dilemma \cite{hardin2}. In these types of games, selfishness is the superior strategy or Nash equilibrium \cite{hardin2}. 
While it can predict winning strategies, however, game theory does not generally consider the feedback of individual behavior phenotypes on group productivity.

According to \cite{rankin}, the exploitation of different types of resources can give rise to a tragedy of the commons. The first, 
which fits Hardin's analogy described above, involves the selfish exploitation of a common, extrinsic resource to the point 
of complete depletion, which causes all individuals to perish. The second type involves resources that are themselves the product of social behavior. In this case, the resource is a public good that is either formed by cooperation, or by restraining from conflict.  Cooperation via public goods is pervasive in microbial social behavior, and it is also the case that we have described here with our model $(\ref{s1})-(\ref{s5})$.  As we have seen, the tragedy arises when non-cooperating cheaters  reap the benefits provided by cooperators, 
without paying the costs. Microbial cooperative behaviors vulnerable to cheating include extracellular secretions such as enzymes  and metabolites \cite{gore,maclean,rainey}. A particularly compelling example is the altruistic investment in the non-spore parts of a 
multicellular fruiting body in myxobacteria \cite{fiegna}.

A different, more abstract, type of public good involves individuals restraining from potential conflict. 
A tragedy arises if the costs invested in compettitive behavior decrease overall productivity. In this case, less emphasis is placed on the depletion of extrinsic resources. A relevant example comes from another chemostat study which  investigated the outcome of social conflict between different metabolic strategies in yeast, respiration and fermentation \cite{maclean}. Respirers use glucose slowly but efficiently, whereas fermenters use glucose fast but wastefully. Thus, respiration is the strategy that provides the highest group-level benefit. Nevertheless, as shown experimentally and confirmed by simulation, fermenters are favored and fully displace respirers during glucose-limited growth in a chemostat \cite{maclean}. Notably, in this system, as in restraint from conflict in general, one strategy does not obligately depend on the other for its success.

As we have proven in this study, population collapse is inevitable in an obligate relationship, because the cooperator to cheater ratio always decreases. Eventually the cheater becomes so dominant that too little public good is produced by the cooperator, leading to the extinction of both types. The differential equation framework presented here will permit the in-depth analysis of mechanisms that promote cooperation.  The contribution of specific parameters or functional forms can be investigated.  For example, how much higher would growth yield or nutrient uptake rates have to be in a cooperator compared to a cheater to make public good cooperation sustainable?

\newpage
\section*{Appendix}
\subsection*{Proof of Theorem $\ref{main}$}
By scaling the state variables of system $(\ref{s1})-(\ref{s5})$ as follows:
\begin{eqnarray*}
s&=&S\\
p&=&P\\
e&=&\frac{E}{\gamma}\\
x_1&=&\frac{X_1}{\gamma}\\
x_2&=&\frac{X_2}{\gamma},
\end{eqnarray*}
and introducing the rescaled functions
\begin{eqnarray*}
g(e,s)&:=&G(\gamma e,s)\\
f(p)&:=&F(P),
\end{eqnarray*}
we obtain the following scaled model:
\begin{eqnarray}
\frac{ds}{dt}(t)&=&D(t)(S^0(t)-s)-g(e,s)\label{red1}\\
\frac{dp}{dt}(t)&=&g(e,s)-(x_1+x_2)f(p)-D(t)p\label{red2}\\
\frac{de}{dt}(t)&=&(1-q)x_1f(p)-D(t)e\label{red3}\\
\frac{dx_1}{dt}(t)&=&x_1\left(qf(p)-D(t) \right)\label{red4}\\
\frac{dx_2}{dt}(t)&=&x_2\left(f(p)-D(t)\right)\label{red5}
\end{eqnarray}

Notice that {\bf H1}, which holds for the rate functions $G(E,S)$ and $F(P)$, is also valid for the scaled rate functions $g(e,s)$ and $f(p)$.

The total mass of this scaled model,
$$
m=s+p+e+x_1+x_2,
$$
satisfies a linear equation:
\begin{equation}\label{mass}
\frac{dm}{dt}(t)=D(t)(S^0(t)-m),
\end{equation}
which is easily verified by adding all the equations of the scaled model. This equation, and the upper bound for $S^0(t)$ in  {\bf H2} imply that
the following family of compact sets
$$
\Omega_\epsilon=\{(s,p,e,x_1,x_2)\, |\,s\geq 0, p\geq 0, e\geq 0, x_1 \geq 0, x_2 \geq 0, m \leq  \bar S^0+\epsilon\},
$$
are forward invariant sets of the scaled model, for all $\epsilon\geq 0$.

The Main Result, Theorem $\ref{main}$, is an immediate Corollary of the following result, which is the
tragedy of the commons for the scaled model:
\begin{stel}\label{scaled-tragedy}Assume that {\bf H1} and {\bf H2} hold, and assume that
the initial condition of $(\ref{red1})-(\ref{red5})$ is such that $x_2(0)>0$; that is, the cheater is present initially. Then
$(p(t),e(t),x_1(t),x_2(t))\to (0,0,0,0)$ as $t\to \infty$.
\end{stel}
\begin{proof}
Given the initial condition, we can find an $\epsilon\geq 0$ such that
the solution $(s(t),p(t),e(t),x_1(t),x_2(t))$ is contained in the compact set $\Omega_{\epsilon}$ for all $t\geq 0$.
We shall present two proofs. The first involves a (biologically nontrivial) transformation of one of the system's variables. The second considers the ratio of cooperators and cheaters, a biologically natural measure, and reveals that this ratio does not increase. \\[2ex]
{\bf Proof 1}: Consider the variable $y_2=x_2^q$. Then
$$
\frac{dy_2}{dt}(t)=y_2(qf(p)-qD(t))
$$
Equation $(\ref{red4})$, and the above equation can be integrated:
\begin{eqnarray*}
x_1(t)&=&x_1(0)\exponent^{\int_0^t qf(p(\tau))-D(\tau)d \tau}\\
y_2(t)&=&y_2(0)\exponent^{\int_0^t qf(p(\tau))-qD(\tau)d \tau}>0,\textrm{ for all } t \textrm{ since } y_2(0)=x_2^q(0)>0,
\end{eqnarray*}
Dividing the first by the second equation yields:
$$
x_1(t)=y_2(t)\frac{x_1(0)}{y_2(0)}\exponent^{-(1-q)\int_0^t D(\tau)d \tau }\leq B\frac{x_1(0)}{y_2(0)}\exponent^{-(1-q){\underline D}t},
$$
where we have used the lower bound for $D(t)$, see {\bf H2}, to establish the last inequality, and the positive bound $B$ for $y_2(t)$ which exists because the solution, and therefore also $x_2(t)$, is bounded. From this follows that
$\lim_{t\to \infty} x_1(t)=0$, where the convergence is at least exponential with rate $(1-q){\underline D}$.\\
Next we consider the dynamics of the variable $z=Qx_1-e$, where $Q=(1-q)/q$:
$$
{\dot z}=-D(t)z,
$$
which is solvable, yielding $z(t)=z(0)\exponent^{-\int_0^t D(\tau) d \tau}$. 
The lower bound ${\underline D}$ for $D(t)$ in {\bf H2}, then implies that $z(t)\to 0$ at a rate which is 
at least exponential with rate ${\underline D}$. This fact, together with the convergence of $x_1(t)$ to zero established above, implies that $e(t)\to 0$ as well.\\
Next, consider the $p$-equation $(\ref{red2})$. There holds that for each ${\tilde \epsilon}> 0$:
$$
\frac{dp}{dt}(t)\leq {\tilde \epsilon}-{\underbar D}p,\textrm{ for all sufficiently large } t.
$$
Notice that we used that $g(0,s)=0$ for all $s\geq 0$, and the continuity of $g$, see {\bf H1}, as well as {\bf H2} for the lower bound of $D(t)$.
It follows that $\limsup_{t\to \infty} p(t)\leq {\tilde \epsilon}/{\underbar D}$, and since ${\tilde \epsilon}>0$ was arbitrary, there follows that $p(t)\to 0$.\\
Finally, we consider the $x_2$-equation $(\ref{red5})$. Since $p(t)\to 0$ and $f(0)=0$ by {\bf H1}, there holds that $f(p(t))\leq {\underbar D}/2$ for all $t$ sufficiently large. Consequently,
$$
\frac{dx_2}{dt}(t)\leq-\frac{{\underbar D}}{2} x_2,\textrm{ for all sufficiently large } t,
$$
and thus $x_2(t)\to 0$, concluding the proof in this case.\\[2ex]
{\bf Proof 2}:
Equations $(\ref{red4})$ and $(\ref{red5})$ can be integrated:
\begin{eqnarray}
x_1(t)&=&x_1(0)\exponent^{\int_0^t qf(p(\tau))-D(\tau)d \tau} \label{x1-sol}\\
x_2(t)&=&x_2(0)\exponent^{\int_0^t f(p(\tau))-D(\tau)d \tau}>0,\textrm{ for all } t\textrm{ since } x_2(0)>0.\label{x2-sol}
\end{eqnarray}
Thus, the ratio $r(t)=x_1(t)/x_2(t)$ is well-defined and satisfies the differential equation:
$$
\frac{dr}{dt}(t)=-(1-q)f(p)r,
$$
which shows that the ratio does not increase. The solution of this equation is:
\begin{equation}\label{r-sol}
r(t)=r(0)\exponent^{-(1-q)\int_0^tf(p(\tau))d \tau}
\end{equation}
We distinguish two cases depending on the integrability of the function $f(p(t))$:

{\bf Case 1}: $\int_0^\infty f(p(\tau))d\tau=\infty$.\\
It follows from  $(\ref{r-sol})$ that $r(t)\to 0$, and hence also $x_1(t)\to 0$ because $x_2(t)$ is bounded.  Proof of convergence of $e(t), p(t)$ and $x_2(t)$ to zero now proceeds as in {\bf Proof 1}.

{\bf Case 2}: $\int_0^\infty f(p(\tau))d\tau<\infty$.\\
It follows from $(\ref{x1-sol})-(\ref{x2-sol})$ that both $x_1(t)\to 0$ and $x_2(t)\to 0$, because 
$0<{\underline D}\leq D(t)$ for all $t$, by {\bf H2}. Proof of convergence of $e(t)$ and $p(t)$ to zero now proceeds as in {\bf Proof 1} as well.\\
\end{proof}

\subsection*{Proof of Theorem $\ref{bistab-unscaled}$}
By scaling the state variables of system $(\ref{c1})-(\ref{c4})$ in the usual way as follows:
\begin{eqnarray*}
s&=&S\\
p&=&P\\
e&=&\frac{E}{\gamma}\\
x_1&=&\frac{X_1}{\gamma},
\end{eqnarray*}
and switching to lower case letters for the rate functions:
\begin{eqnarray*}
g(s)&:=&G(S)\\
f(p)&:=&F(P),
\end{eqnarray*}
and for the chemostat's constant operating parameters:
\begin{eqnarray*}
d&:=&D\\
s^0&:=&S^0,
\end{eqnarray*}
we obtain the following scaled model:
\begin{eqnarray}
\frac{ds}{dt}(t)&=&d(s^0-s)-eg(s)\label{cs1}\\
\frac{dp}{dt}(t)&=&eg(s)-x_1f(p)-dp\label{cs2}\\
\frac{de}{dt}(t)&=&(1-q)x_1f(p)-de\label{cs3}\\
\frac{dx_1}{dt}(t)&=&x_1\left(qf(p)-d \right)\label{cs4}
\end{eqnarray}
Notice that {\bf H1'}, which holds for the rate functions $G(S)$ and $F(P)$, is also valid for the rate functions $g(s)$ and $f(p)$.

We introduce two new variables:
\begin{eqnarray}
m&=&s+p+e+x_1\label{newvar1}\\
v&=&e-Qx_1,\textrm{ where }Q:=\frac{1-q}{q}\label{newvar2}
\end{eqnarray}
and choose to drop the $s$ and $e$-equations from system $(\ref{cs1})-(\ref{cs4})$, transforming it to:
\begin{eqnarray}
\frac{dm}{dt}(t)&=&d(s^0-m) \label{t1}\\
\frac{dv}{dt}(t)&=&-dv\label{t2}\\
\frac{dp}{dt}(t)&=&(v+Qx_1)g(m-p-v-x_1/q)-xf(p)-dp\label{t3}\\
\frac{dx_1}{dt}(t)&=&x_1\left(qf(p)-d \right)\label{t4}
\end{eqnarray}
with state space $\{p,x_1\geq 0\,:\, m\geq p+v+x_1/q,\, v+Qx_1\geq 0\}$, which is forward invariant.  The variables $m(t)$ and $v(t)$ converge exponentially to
$s^0$ and $0$ respectively, hence it is natural to study the limiting system:
\begin{eqnarray}
\frac{dp}{dt}(t)&=&Qx_1g(s^0-p-x_1/q)-x_1f(p)-dp\label{l1}\\
\frac{dx_1}{dt}(t)&=&x_1(qf(p)-d)\label{l2}
\end{eqnarray}
which is defined on the state space $\{p,x_1\geq 0\,:\, p+x_1/q\leq s^0\}$, which is forward invariant. It turns out to be more convenient to transform this system using the variable:
\begin{equation}\label{w-var}
w=p+\frac{1}{q}x_1,
\end{equation}
instead of the variable $p$, yielding:
\begin{eqnarray}
\frac{dw}{dt}(t)&=&Qx_1g(s^0-w)-dw \label{lt1}\\
\frac{dx_1}{dt}(t)&=&x_1(qf(w-x_1/q)-d) \label{lt2}
\end{eqnarray}
with state space $\Omega_{red}=\{x_1\geq 0:\, x_1/q\leq w\leq s^0\}$, which is forward invariant.

We start our analysis of system $(\ref{lt1})-(\ref{lt2})$ by determining the nullclines. The $w$-nullcline is given by:
\begin{equation}\label{null-w}
x_1=h(w),\textrm{ where }h(w)=\frac{d}{Q}\frac{w}{g(s^0-w)}.
\end{equation}
The main properties of the function $h(w):[0,s^0)\to \reals_+$ are:
\begin{enumerate}
\item $h(0)=0,\textrm{ and }\lim_{w\to s^0}h(w)=\infty$.
\item $h'(w)=\frac{d}{Q}\frac{g(s^0-w)+wg'(s^0-w)}{g^2(s^0-w)}>0$, by {\bf H1'}.
\item $h''(w)=\frac{d}{Q}\frac{-wg''(s^0-w)g^2(s^0-w)+2g(s^0-w)g'(s^0-w)\left(g(s^0-w)+wg'(s^0-w) \right)}{g^4(s^0-w)}>0$, by ${\bf H1'}$.
\end{enumerate}
Thus, the function $h(w)$ is zero at zero, is increasing with a vertical asymptote at $w=s^0$,  and it is strictly convex.

To obtain a nontrivial $x_1$-nullcline in the state space, we make one more assumption, namely:
\begin{equation}\label{pstar}
p^*:=f^{-1}\left(\frac{d}{q}\right)\textrm{ satisfies } p^*< s^0.
\end{equation}
This assumption merely expresses that the cooperator has a break-even steady state concentration for the processed nutrient at a level
below the input nutrient concentration $s^0$, see equation $(\ref{cs4})$. In addition to the horizontal axis $x_1=0$, there is a nontrivial
$x_1$-nullcline which is particularly easy to express using $p^*$, as the graph of a linear function:
\begin{equation}\label{null-x_!}
x_1=q(w-p^*).
\end{equation}
Any nonzero steady states of the limiting system are given by the intersection of the $w$- and the nontrivial $x_1$-nullcine, which are determined by the solutions of the equation:
\begin{equation}\label{intersect}
h(w)=q(w-p^*),\; 0\leq w<s^0.
\end{equation}
In view of the convexity of the function $h$, there are either no, one, or two solutions to $(\ref{intersect})$, and generically there are none, or two. We will construct the phase portrait of the limiting system in these two cases.
\begin{figure}[h]
\begin{center}
\includegraphics[width=6in]{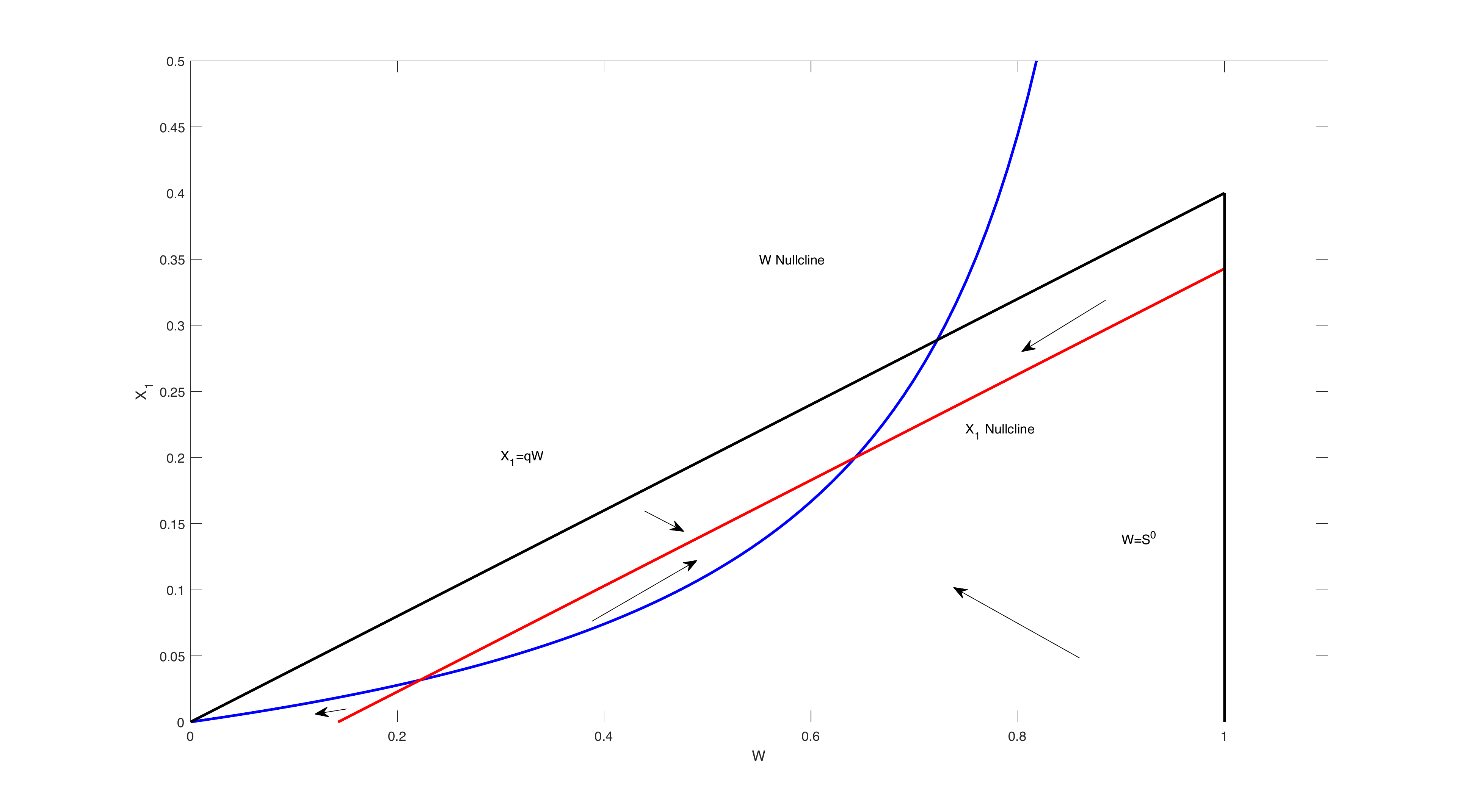}
\caption{Phase plane of the limiting system $(\ref{lt1})-(\ref{lt2})$ in case there are two nonzero steady states.}
\label{phase}
\end{center}
\end{figure}
\begin{lemma}\label{planar}
Suppose that {\bf H1'} and $(\ref{pstar})$ hold.
\begin{enumerate}
\item
If equation $(\ref{intersect})$ has no solutions, then system $(\ref{lt1})-(\ref{lt2})$ has a unique steady state $(0,0)$ which is globally asymptotically stable with respect to initial conditions in $\Omega_{red}$.
\item
If equation $(\ref{intersect})$ has two solutions $w_1<w_2$, then system $(\ref{lt1})-(\ref{lt2})$ has 3 steady states, $(0,0)$, $(w_1,h(w_1))$ and
$(w_2,h(w_2))$. The steady states $(0,0)$ and $(w_2,h(w_2))$ are locally asymptotically stable, and $(w_1,h(w_1))$ is a saddle with one-dimensional stable manifold $W_s$, and one-dimensional unstable manifold $W_u$. The stable manifold $W_s$ intersects the boundary of $\Omega_{red}$ in two points, one on the boundary $x_1=qw$, the other on the boundary $w=s^0$, forming a separatrix: Initial conditions below $W_s$ give rise to solutions converging to $(0,0)$, whereas
initial conditions above $W_s$ give rise to solutions converging to $(w_2,h(w_2))$, yielding bistability in the limiting system, see Figure $\ref{phase}$.
\end{enumerate}
\end{lemma}
\begin{proof}
\begin{enumerate}
\item If $(\ref{intersect})$ has no solutions, then the $w$-nullcline and the nontrivial $x_1$-nullcline do not intersect, and thus $(0,0)$ is the only steady state of the system. The state space $\Omega_{red}$ is divided in 3 parts by the two nullclines, and it is easy to see that the region enclosed between both nullclines and the boundary of $\Omega_{red}$ is
a trapping region in which solutions monotonically converge to the zero steady state. Solutions starting in the region above the $w$-nullcline are monotonically decreasing (increasing) in the $x_1$-component ($w$-component), but since that region does not contain nontrivial steady states, these solutions must enter the trapping region between both nullclines. Similarly, solutions that start below the $x_1$-nullcline are monotonically increasing (decreasing) in the $x_1$-component ($w$-component), and must enter the trapping region as well. This concludes the proof of the assertion that $(0,0)$ is a globally asymptotically stable steady state.
\item If equation $(\ref{intersect})$ has two solutions $w_1<w_2$, then the nullclines intersect in two distinct points, yielding the positive
steady states $(w_1,h(w_1))$ and $(w_2,h(w_2))$, see Figure $\ref{phase}$. The third steady state is $(0,0)$. It is not hard to see that
the state space is now divided into 5 parts, 3 of which are trapping regions. Each of these trapping regions is enclosed by arcs of the nullclines or segments of the boundary of $\Omega_{red}$ which either connect pairs of steady states, or a steady state and a point on the boundary of the state space $\Omega_{red}$, see
Figure $\ref{phase}$. There are also 2 remaining regions, which we call the NW and SE regions, for obvious reasons.

To complete the phase plane analysis, we perform a linearization of the system at the steady states. The Jacobian matrix of the limiting system is
$$
\begin{pmatrix}
-Qx_1g'(s^0-w)-d& Qg(s^0-w)\\
x_1qf'(w-x_1/q)& (qf(w-x_1/q)-d)-x_1f'(w-x_1/q)
\end{pmatrix}
$$
We focus on the middle steady state $(w_1,h(w_1))$, where the Jacobian evaluates to:
$$
J_1=\begin{pmatrix}
-d\left(\frac{w_1g'(s^0-w_1)}{g(s^0-w_1)}+1 \right)& Qg(s^0-w)\\
\frac{d}{Q}q\frac{w_1}{g(s^0-w_1)}f'(p^*)& -\frac{d}{Q}\frac{w_1}{g(s^0-w_1)}f'(p^*)
\end{pmatrix}
$$
Clearly, the trace is negative, and the determinant is given by:
$$
\frac{d}{Q}\frac{w_1}{g(s^0-w_1)}f'(p^*)\left[d\left(\frac{w_1g'(s^0-w_1)}{g(s^0-w_1)}+1 \right)-qQg(s^0-w_1) \right]
$$
We claim that this determinant is negative, which implies that this steady state is a saddle. This follows from the fact that the slope of the tangent line to the graph of the convex function $h(w)$ at $w=w_1$ must be smaller than the slope of the line $x_1=q(w-p^*)$, which is of course $q$:
$$
h'(w_1)<q.
$$
Recalling the derivative of $h(w)$ given above, it can be shown that this latter inequality is equivalent to the expression in the square brackets in the determinant being negative, which establishes the claim.

Incidentally, a similar argument can be used to show that the determinant of the linearization at the
steady state $(w_2,h(w_2))$ is positive, because in this case $h'(w_2)>q$. Since the trace of that linearization is also negative, this shows that $(w_2,h(w_2))$ is locally asymptotically stable. The linearization at $(0,0)$ is triangular with both diagonal entries equal to $-d$, from which also follows that $(0,0)$ is locally asymptotically stable.

Now we turn to the question of the location of the one-dimensional stable and unstable manifolds $W_s$, respectively $W_u$, of the saddle.
Therefore, we determine the eigenvectors of the negative eigenvalue $\lambda_1$, and the positive eigenvalue $\lambda_2$ of the  Jacobian matrix $J_1$. We have that
$$
J_1\begin{pmatrix}1\\ r_1 \end{pmatrix}=\lambda_1\begin{pmatrix}1\\ r_1 \end{pmatrix}\textrm{ with }\lambda_1<0,
\textrm{ and }
J_1\begin{pmatrix}1\\ r_2 \end{pmatrix}=\lambda_2\begin{pmatrix}1\\ r_2 \end{pmatrix}, \textrm{ with }\lambda_2>0,
$$
where we wish to determine, or at least estimate, $r_1$ and $r_2$. We can find $r_1$ by considering the first of the two equations
determining $\lambda_1$:
$$
r_1=\frac{1}{g(s^0-w)}\left(\lambda_1+d\left(\frac{w_1g'(s^0-w_1)}{g(s^0-w_1)}+1 \right) \right).
$$
We claim that $r_1<0$. Indeed, since the trace of $J_1$ (which equals $\lambda_1+\lambda_2$) is less than the top-left entry of $J_1$, it follows that the expression in the large parentheses is less than $-\lambda_2$, which is negative. This implies that near $(w_1,h(w_1))$, the stable manifold has a branch in the NW, and another branch in the SE region. Backward integration of solutions starting near the saddle and on $W_s$ in these regions, shows that they must either exit these regions along the boundary of $\Omega_{red}$ (because backward-time solutions cannot exit via the trapping regions), or they must converge to a steady state. However, there are no steady states to the NW of the saddle in the NW region, nor to the
SE of the saddle in the SE region. Thus, the stable manifold $W_s$ must intersect the boundary of $\Omega_{red}$ in two points, one on the line $x_1=qw$, and the other on the line $w=s^0$. Next, we focus on the location of the unstable manifold $W_u$ of the saddle, whose location is determined by $r_2$. We claim that:
$$
h'(w_1)<r_2<q.
$$
To see this we consider the first equation in the eigenvalue equation for $\lambda_2$, by solving for $r_2$, once again recalling the
expression of the derivative of the function $h$:
$$
r_2=\frac{\lambda_2}{Qg(s^0-w_1)}+h'(w_1),
$$
from which the first inequality follows because $\lambda_2>0$. The second equation in the eigenvalue equation yields that:
$$
r_2=q\frac{h(w_1)f'(p^*)}{h(w_1)f'(p^*)+\lambda_2},
$$
and again, since $\lambda_2>0$, we find that the second inequality holds, as claimed.

We can now fully assemble the phase portrait presented in Figure $\ref{phase}$. The stable manifold $W_s$ of the saddle has a branch in the
first and second trapping regions. Solutions starting on these branches must converge to $(0,0)$, and $(w_2,h(w_2))$ respectively.
In fact, it is not hard to see that all solutions in the first trapping region converge to $(0,0)$, whereas solutions in the second and third trapping region converge to $(w_2,h(w_2))$. The fate of solutions starting in the NW and SE regions depends on their initial location relative to
the separatrix $W_s$: They converge to $(0,0)$ if they start  below $W_s$, but to $(w_2,h(w_2))$ if they start above $W_s$, and this occurs because they must enter one of the trapping regions first.
\end{enumerate}
\end{proof}

The asymptotic behavior of the solutions of system $(\ref{lt1})-(\ref{lt2})$ described in Lemma $\ref{planar}$ can be translated into the asymptotic behavior of the solutions of system $(\ref{l1})-(\ref{l2})$, and combining this
with the theory of asymptotically autonomous systems, see Appendix F in \cite{chemostat}, the asymptotic behavior of the transformed system $(\ref{t1})-(\ref{t4})$ can be obtained as well. In turn, this determines the behavior of the scaled system $(\ref{cs1})-(\ref{cs4})$, from which Theorem $\ref{bistab-unscaled}$ follows immediately.  

\subsection*{More general digestion networks}
To see why the tragedy continues to hold for more general digestion network, we consider 
solutions of system $(\ref{bio1})-(\ref{bio6})$, for which it is easily verified that the variable:
$$
m=s+p+e+2c+x_1+x_2,
$$
still satisfies equation $(\ref{mass})$, implying that the family of compact sets $\Omega_\epsilon$, defined earlier, is forward
invariant for system $(\ref{bio1})-(\ref{bio6})$, for all $\epsilon\geq 0$, when {\bf H2} holds. Consequently, the proof of Theorem $\ref{scaled-tragedy}$ remains valid for
the above chemostat model $(\ref{bio1})-(\ref{bio6})$. Indeed, the first proof only crucially depends on the dynamics of $x_1$ and $x_2$ to show that $x_1(t)$ converges to zero, after which the convergence of $e$, $p$ and $x_2$ is obtained by elementary comparison arguments.
For the digestion network presented here, the dynamics of $x_1$ and $x_2$ remain unchanged, hence we can still conclude that $x_1(t)$ converges to zero. After that, it follows from a comparison argument that $e+c$ converges to zero, and then similarly that $p$ and $x_2$ converge to zero as well. One could also easily adapt the steps of the second proof to obtain the same conclusion.

\end{document}